\documentclass[a4paper,11pt]{article}

\usepackage{amsfonts}
\usepackage{amsmath}
\usepackage{amssymb}
\usepackage{url}
\newcommand{\F}{\mathbb{F}}

\newtheorem{theorem}{Theorem}
\newtheorem{lemma}[theorem]{Lemma}
\newtheorem{proposition}{Proposition}
\newtheorem{corollary}[theorem]{Corollary}

\newtheorem{conjecture}{Conjecture}
\newtheorem{remark}{Remark}
\newenvironment{proof}{\noindent {\em Proof.}}{\hspace*{\fill} $\Box $\newline}

\begin{document}

\title{Optimal Binary LCD Codes}

\author{%Iliya Bouyukliev\\
%Institute of Mathematics and Informatics\\
%Bulgarian Academy of Sciences\\
%P.O.Box 323, 5000 Veliko Tarnovo, Bulgaria\\
%iliya@moi.math.bas.bg
%\and
Stefka Bouyuklieva\\
Faculty of Mathematics and Informatics,\\
              St. Cyril and St. Methodius University,\\
Veliko Tarnovo, Bulgaria\\
stefka@ts.uni-vt.bg
}

\date{}

\footnotetext[1]{This research was
supported in part by by a Bulgarian NSF contract KP-06-N32/2-2019.}

\maketitle

%%%%%%%%%%%%%%%%%%  abstract  %%%%%%%%%%%%%%%%

%\hfill \\
\begin{abstract}
Linear complementary dual codes (shortly LCD codes) are codes whose intersections with their dual codes are trivial. These codes were first introduced by Massey in 1992. Nowadays, LCD codes are extensively studied in the literature and widely applied in data storage, cryptography, etc. In this paper, we prove some properties of binary LCD codes using their shortened and punctured codes. We also present some inequalities for the largest minimum weight $d_{LCD}(n,k)$ of binary LCD $[n,k]$ codes for given length $n$ and dimension $k$. Furthermore, we give two tables with the values of $d_{LCD}(n,k)$ for $k\le 32$ and $n\le 40$, and two tables with classification results.
\end{abstract}

\vspace{3mm}

\section{Introduction}
\label{intro}

A code $C$ is a linear complementary dual (LCD) code if $C\cap C^\perp=\{0\}$.
LCD codes were introduced by Massey \cite{Massey} and gave an optimum linear
coding solution for the two user binary adder channel. The classification of LCD $[n, k]$ codes and determination of the largest minimum weight among all LCD $[n, k]$ codes are fundamental problems. Recently, much work has been done concerning this
fundamental problem (see \cite{Carlet_Pellikaan,Dougherty-Kim-Sole-LCD,Kim-LCD,Harada-LCD}). It has been
shown in \cite{Carlet_Pellikaan} that any code over $\mathbb{F}_q$ is equivalent to some LCD code for $q \ge 4$.  The hull of a linear code is defined to be its intersection with its dual. A code $C$ coincides with its hull if and only if $C$ is self-orthogonal. The LCD codes have the opposite property with respect to the hulls - a linear code is LCD if and only if its hull consists only of the zero vector.

A linear $[n,k]$ code is {\em optimal} if it has the largest
minimum distance among all linear $[n,k]$ codes
(see \cite{Grassl-table} for known bounds on the minimum distance of linear codes).
For some values of
$n$ and $k$ no optimal linear $[n,k]$ code is LCD. We denote by $d_{LCD}(n, k)$ the maximum value of $d$ for which an LCD $[n, k, d]$ code exists. We call the LCD codes with parameters $[n,k,d_{LCD}(n, k)]$ {\em optimal LCD codes}. Binary optimal LCD
codes have been classified up to length 16 in
\cite{Harada-LCD}. In \cite{Araya-Harada}, all values of $d_{LCD}(n, k)$ for $n\le 24$ are obtained and some classification results on LCD codes with such lengths are presented. We extend the tables with values of $d_{LCD}(n, k)$ to length 40. Construction methods that allow to obtain effective LCD binary codes with large minimum
distance and large rate have been proposed in \cite{Carlet_Guilley}.

The paper is organized as follows. Section \ref{sect:preliminaries} contains some basic facts and definitions about linear codes over a finite field. In Section \ref{sect:properties} we give some necessary theorems and prove some properties of binary LCD codes connected with their shortened and punctured codes. Moreover, we obtain as corollaries inequalities for  $d_{LCD}(n, k)$. We apply these inequalities to obtain values of $d_{LCD}(n, k)$ for given $n$ and $k$ and put the results in two tables. Table \ref{table1} contains values of $d_{LCD}(n, k)$ for $n=16,\ldots,40$ and $k=5,\ldots,17$, and Table \ref{table2} - for $n=23,\ldots,40$ and $k=18,\ldots,32$. If the exact value of $d_{LCD}(n,k)$ is not obtained, we give an interval, for example 8--10 means that $8\le d_{LCD}(n, k)\le 10$. The results on $d_{LCD}(n, k)$ are included in Section \ref{sect:dLCD}. We present some classification results in Section \ref{sect:classification}. To classify the codes, we use the program \textsc{Generation} \cite{Generation} with a modification for LCD codes. We complete the paper with a brief conclusion.

\section{Preliminaries}
\label{sect:preliminaries}

Let $\F_q$ be a finite field with $q$ elements and $\F_q^n$ be the $n$-dimensional vector space over $\F_q$. The (Hamming) \textit{distance} $d(x,y)$ between two vectors $x,y\in\F_q^n$ is the number of coordinate positions in which they differ. The (Hamming) \textit{weight} wt$(x)$ of a vector $x\in\F_q^n$ is the number of its nonzero coordinates.
  A linear $[n,k,d]$ code $C$ is a $k$-dimensional subspace of the
vector space $\mathbb{F}_q^n$, and $d$ is the
smallest weight among all non-zero codewords of $C$, called the \textit{minimum weight} (or minimum distance) of the code. A matrix whose rows form a basis of $C$
is called a generator matrix of this code. The weight enumerator
$W(y)$ of a code $C$ is given by $W(y)=\sum_{i=0}^n A_iy^i$ where
$A_i$ is the number of codewords of weight $i$ in $C$. A code is called {\em even-like} (or just even) if the weights of all codewords are
even and {\em odd-like} otherwise.

Two binary codes $C$ and
$C'$ are equivalent if one can be obtained from the other by
permuting the coordinates. The automorphism group of $C$ is the
set of permutations of the coordinates which preserve the code.

The dual code $C^\perp$ of a binary linear code $C$ is defined as $C^\perp=\{x \in
\F_2^n : x \cdot y=0 \mbox{ for all } y\in C\}$ where $x \cdot
y$ denotes the standard inner product of $x$ and $y$. If a linear code $C$ is equivalent to its orthogonal complement $C^\perp$, then it is termed \textit{isodual}, and if $C=C^\perp$, $C$ is a \textit{self-dual} code. A linear
code $C$ is formally self-dual if $C$ and its dual code
$C^{\perp}$ have the same weight enumerator. Any isodual code is also formally self-dual, but there are formally self-dual codes that are not isodual. It is interesting to note that there are LCD codes that are isodual. For example, the only binary $[8,4,3]$ LCD code is equivalent to its dual code.

Different methods have been used to study, construct and classify LCD codes with different parameters over different finite fields, for example $k$-covers of $m$-sets \cite{Harada-LCD}, principal submatrices \cite{Kim-LCD}, etc.
The following characterization is due to Massey \cite{Massey}.

\begin{theorem}\label{thm:Massey}
 Let $G$ and $H$ be a generator matrix and
a parity-check matrix of a code $C$, respectively. Then the following properties are
equivalent:

(i) $C$ is LCD,

(ii) $C^\perp$ is LCD,

(iii) $GG^T$ is nonsingular,

(iv) $HH^T$ is nonsingular.

\end{theorem}

Another important property is given in the following corollary.

\begin{corollary}\label{cor:sum}
The linear code $C$ is LCD if and only if $\F_q^n=C\oplus C^\perp$.
\end{corollary}

We prove some properties of LCD codes using their shortened and punctured codes on one coordinate position.
Let $C$ be a linear code of length $n$ and $i\in\{1,2,\ldots,n\}$. Denote by $C^i$ the punctured code, and by $C_i$ the shortened code of $C$ on the $i$-th coordinate. We need the following result on the punctured and shortened codes of $C$ that is a modification of \cite[Theorem 1.5.7]{HP}.

\begin{lemma}
\label{shortened-punctured} Let $C$ be an $[n,k,d]$ code and $i\in\{1,2,\ldots,n\}$. Then:

(i) $(C^\perp)_i=(C^i)^\perp$ and $(C^\perp)^i=(C_i)^\perp$;

(ii) if $d>1$, then $C^i$ and $(C^\perp)_i$ have dimensions
$k$ and $n-1-k$, respectively.
\end{lemma}

\section{Some Properties of the LCD Codes}
\label{sect:properties}

From now on, all codes are assumed to be binary.
Obviously, if $C$ is an LCD code with dual distance 1 and $C=(O\vert C')$ then $C'$ is also LCD. Conversely, if $C'$ is LCD then $C=(O\vert C')$ is also LCD (by $O$ we denote the zero code as well as the zero matrix of the suitable size). Consider an LCD code $C$ with dual distance $d^\perp\ge 2$. Let $C_1$ and $C^1$ be the shortened and the punctured codes of $C$ on the first coordinate, respectively. Then $C^1=C_1\cup (x+C_1)$, where $x\in C^1\setminus C_1$. Moreover, $(1,x)\in C$ and $C=(0\vert C_1)\cup (1\vert x+C_1)$.

\begin{lemma}\label{lemma:p-s}
If $C$ is an LCD code with $d\ge 2$ and $d^\perp\ge 2$, then exactly one of the codes $C_1$ and $C^1$ is also LCD.
\end{lemma}

\begin{proof} Let $\mathcal{B}_1$ and $\mathcal{B}^1$ be the shortened and the punctured codes of $\mathcal{B}=C^\perp$ on the first coordinate, respectively. Then $\mathcal{B}^1=\mathcal{B}_1\cup (y+\mathcal{B}_1)$, where $y\in \mathcal{B}^1\setminus \mathcal{B}_1$. Moreover, $\mathcal{B}^1=C_1^\perp$ and $\mathcal{B}_1=(C^1)^\perp$.

Since $\F_2^n=C\oplus C^\perp$ then $(10\ldots 0)=(0,v)+(1,v)$, where for the vector $v$ there are two possibilities - $v\in C_1$ ($(0,v)\in C$ and $(1,v)\in C^\perp$) or $v\in\mathcal{B}_1$ ($(0,v)\in C^\perp$ and $(1,v)\in C$).
\begin{itemize}
\item Let $v\in C_1$. Then $(1,v)\in C^\perp$, so $v\in\mathcal{B}^1=C_1^\perp$ and therefore $C_1$ is not an LCD code. Suppose that $C^1$ is not LCD either. Let $w\in C^1\cap (C^1)^\perp=C^1\cap \mathcal{B}_1$. Hence $(0,w)\in C^\perp$ and $(1,w)\in C$. It turns out that $(10\ldots 0)=(1,w)+(0,w)$, $(1,w)\in C$, $(0,w)\in C^\perp$. In this way the vector $(10\ldots 0)$ is represented as a sum of a codeword from $C$ plus a codeword from $C^\perp$ in two ways, which contradicts to Corollary \ref{cor:sum}. Hence in this case $C^1$ is an LCD code but $C_1$ is not.
\item Let $v\in \mathcal{B}_1=(C^1)^\perp$. Then $\mathcal{B}_1$ is not an LCD code. Similarly to the previous case, it follows that $\mathcal{B}^1$ is LCD code. Since $\mathcal{B}_1=(C^1)^\perp$ and $\mathcal{B}^1=C_1^\perp$, in this case $C_1$ is an LCD code but $C^1$ is not.
\end{itemize}
This proves that exactly one of the codes $C_1$ and $C^1$ is an LCD code.
\end{proof}

 To prove some more properties of the binary LCD codes, we need the following theorems.

\begin{theorem}\cite{CarletLCD2019}\label{thm:odd}
Let $C$ be an odd-like binary code with parameters $[n,k]$.
Then $C$ is LCD if and only if there exists a basis $c_1,c_2,\ldots,c_k$ of $C$ such that for any $i, j \in\{ 1,2,\ldots,k\}$, $c_i\cdot c_j$ equals 1 if $i = j$ and equals 0 if $i\neq j$.
\end{theorem}

\begin{theorem}\cite{CarletLCD2019}\label{thm:even}
Let $C$ be an even-like binary code with parameters $[n,k]$.
Then $C$ is LCD if and only if $k$ is even and there exists a basis $c_1,c'_1$, $c_2,c'_2,\ldots$, $c_{k/2},c'_{k/2}$ of $C$ such that for any $i, j \in\{ 1,2,\ldots,k/2\}$,
  the following conditions hold

  (i) $c_i\cdot c_i =c'_i\cdot c'_i =0$;

  (ii) $c_i\cdot c'_j =0$ if $i\neq j$;

  (iii) $c_i \cdot c'_i =1$.
\end{theorem}

In the following proposition, we prove some connections between odd and even binary LCD codes.

\begin{proposition}\label{prop:n+1} Let $C$ be an LCD $[n,k,d]$ code for an even $k$.

(i) If $C$ is odd-like and $c_1,c_2,\ldots,c_k$ is a basis of $C$ with the property given in Theorem \ref{thm:odd}, then the vectors $(1,c_1),(1,c_2),\ldots,(1,c_k)$ generate an even-like LCD $[n+1,k,d$ or $d+1]$ code.

(ii) If $C$ is even-like and $c_1,c'_1,c_2,c'_2,\ldots,c_{k/2},c'_{k/2}$ is a basis of $C$ with the property given in Theorem \ref{thm:even}, then the vectors $(1,c_1),(1,c'_1)$, $(1,c_2)$, $(1,c'_2),\ldots$, $(1,c_{k/2})$, $(1,c'_{k/2})$ generate an odd-like LCD $[n+1,k,d$ or $d+1]$ code.
\end{proposition}

\begin{proof} (i) Let $C$ be an odd-like code with a generator matrix $G$ whose rows are the codewords $c_1,c_2,\ldots,c_k$, $c_i\cdot c_j=1$ if $i = j$ and $c_i\cdot c_j=0$ if $i\neq j$. Then $GG^T=I_k$, where $I_k$ is the $k\times k$ identity matrix. If $\overline{G}=(1\vert G)$, then
\[ \overline{G}\overline{G}^T=\left(\begin{array}{cccc}
    0&1&\cdots&1\\
    1&0&\cdots&1\\
    &&\ddots&\\
    1&1&\cdots&0\\
    \end{array}\right)=J_k-I_k,
    \]
    where $J_k$ is the all-ones matrix. Calculating the determinant of this matrix, we have det$(\overline{G}\overline{G}^T)=1-k\equiv 1\pmod{2}$. It turns out that $\overline{G}\overline{G}^T$ is nonsingular and therefore $\overline{C}=\langle \overline{G}\rangle$ is an LCD $[n+1,k,d$ or $d+1]$ code. Since all rows in its generator matrix $\overline{G}$ are vectors of even weights, the code is even.

    (ii) Let $C$ be an even-like code with a generator matrix $G$ whose rows are the codewords $c_1,c'_1,c_2,c'_2,\ldots,c_{k/2},c'_{k/2}$. If $\overline{G}=(1\vert G)$, then
\[ \overline{G}\overline{G}^T=\left(\begin{array}{cccc}
    I_2&J_2&\cdots&J_2\\
    J_2&I_2&\cdots&J_2\\
    &&\ddots&\\
    J_2&J_2&\cdots&I_2\\
    \end{array}\right),
    \]
    where $I_2$ is the identity matrix, and $J_2$ is the all-ones matrix of order 2. The determinant of this matrix is equal to $(-1)^{k/2}(1-k)\equiv 1\pmod 2$. Hence the code, generated by $\overline{G}$ is LCD. The rows of $\overline{G}$ have odd weights and therefore the code is odd-like.
\end{proof}

\begin{remark}\rm A similar statement is given in \cite[Lemma 3.2]{Araya-Harada-Saito}.
\end{remark}

\begin{corollary}\label{cor:k-even-d-odd}
If $k$ is even and $d_{LCD}(n,k)$ is odd, then $d_{LCD}(n+1,k)\ge d_{LCD}(n,k)+1$.
\end{corollary}

The following propositions show some properties of the shortened and punctured codes of an LCD code.

\begin{proposition}\label{prop:punctured} If $C$ is an even-like $[n,k,d]$ LCD code, then the punctured code of $C$ on any coordinate is again LCD.
\end{proposition}

\begin{proof} Let's consider the punctured code $C^1$ of $C$ on the first coordinate.
According to \cite[Lemma 7]{CarletLCD2019},  there exists a basis $c_1,c'_1,c_2,c'_2,\ldots,c_{k/2},c'_{k/2}$ of $C$ such that for any $i, j \in\{ 1,2,\ldots,k/2\}$,
  the following conditions hold: (i) $c_i\cdot c_i =c'_i\cdot c'_i =0$; (ii) $c_i\cdot c'_j =0$ if $i\neq j$; (iii) $c_i \cdot c'_i =1$;
    (iv) $c_{i,1} = c'_{i,1}$, where $c_i = (c_{i,1},\ldots , c_{i,n} )$ and $c'_i = (c'_{i,1},\ldots , c'_{i,n} )$.

  Without loss of generality we can suppose that $c_{i,1} = c'_{i,1}=1$ for $i=1,\ldots,s$, and $c_{i,1} = c'_{i,1}=0$ for $i=s+1,\ldots,k/2$. Denote by $b_i$ and $b'_i$ the vectors $(c_{i,2},\ldots , c_{i,n} )$ and $(c'_{i,2},\ldots , c'_{i,n} )$, respectively. Hence (i) $b_i\cdot b_i =b'_i\cdot b'_i =1$ if $i\le s$ and 0 otherwise, (ii) $b_i\cdot b'_j =1$ if $i\neq j$, $i,j\in\{1,\ldots,s\}$, and 0 otherwise,(iii) $b_i \cdot b'_i =0$ if $i\le s$ and 1 otherwise. It turns out that
  $$ G_1G^T_1=\left(\begin{array}{cc}
  A&O\\
  O&B\\
  \end{array}\right),$$
  where $G_1$ is the matrix with rows $b_1,b'_1,b_2,b'_2,\ldots,b_{k/2},b'_{k/2}$, $A$ and $B$ are $2s\times 2s$ and $(k-2s)\times (k-2s)$ matrices, respectively,
  \[ A=\left(\begin{array}{cccc}
  I_2&J_2&\cdots&J_2\\
  J_2&I_2&\cdots&J_2\\
  \vdots&\vdots&\ddots&\vdots\\
  J_2&J_2&\cdots&I_2
  \end{array}\right), \ \ \ \ B=\left(\begin{array}{cccc}
  J_2-I_2&O&\cdots&O\\
  O&J_2-I_2&\cdots&O\\
  \vdots&\vdots&\ddots&\vdots\\
  O&O&\cdots&J_2-I_2
  \end{array}\right).\]
  The determinant of $G_1G^T_1$ is equal to $(-1)^{k/2-1}(2s-1)\equiv 1\pmod 2$ and hence this matrix is nonsingular over $\F_2$. It follows that the punctured code $C^1$ is an LCD $[n-1,k,d$ or $d-1]$ code.
  %\[ det(A)=\left|\begin{array}{cccc}
%  I_2&J_2&\cdots&J_2\\
%  J_2&I_2&\cdots&J_2\\
%  \vdots&\vdots&\ddots&\vdots\\
%  J_2&J_2&\cdots&I_2
%  \end{array}\right|=\left|\begin{array}{cccc}
%  I_2&J_2&\cdots&J_2\\
%  J_2-I_2&I_2-J_2&\cdots&O\\
%  \vdots&\vdots&\ddots&\vdots\\
%  J_2-I_2&O&\cdots&I_2-J_2
%  \end{array}\right| \]
%  \[=\left|\begin{array}{cccc}
%  I_2+(s-1)J_2&J_2&\cdots&J_2\\
%  O&I_2-J_2&\cdots&O\\
%  \vdots&\vdots&\ddots&\vdots\\
%  O&O&\cdots&I_2-J_2
%  \end{array}\right|=|I_2+(s-1)J_2||I_2-J_2|^{s-1}
%  \]
%  $$=\left|\begin{array}{cc}
%  s&s-1\\ s-1&s\end{array}\right| \left|\begin{array}{cc}
%  0&-1\\ -1&0\end{array}\right|^{s-1}=(s^2-s^2+2s-1)(-1)^{s-1}=(2s-1)(-1)^{s-1}$$
%  $$det B=|J_2-I_2|^{k/2-s}=(-1)^{k/2-s}$$
\end{proof}

\begin{proposition}\label{prop:odd-like}
Let $C$ be an odd-like $[n,k,d]$ LCD code. If $C$ contains the all-ones vector, then the shortened code of $C$ on any coordinate is again LCD. If $C$ does not contain the all-ones vector, then there are integers $i,j\in\{1,\ldots,n\}$ such that the shortened code of $C$ on the $i$-th coordinate and the punctured code of $C$ on the $j$-th coordinate are LCD codes.
\end{proposition}

\begin{proof} If $C$ contains the all-ones vector, then its dual code $C^\perp$ is even. According to the previous proposition, the punctured code of $C^\perp$ on any coordinate is again LCD code, hence the shortened code of $C$ on any coordinate is again LCD.

Now suppose that $C$ does not contain the all-ones vector.
According to Theorem \ref{thm:odd}, there is a basis $c_1,\ldots,c_k$ of $C$ such that $c_i\cdot c_i=1$ for $i=1,\ldots,k$, and $c_i\cdot c_j=0$ for $i\neq j$. Suppose that $c_i=(1,b_i)$ for $i=1,\ldots,s$, and $c_i=(0,b_i)$ otherwise. Then $b_i\cdot b_i=0$ for $i=1,\ldots, s$, $b_i\cdot b_i=1$ for $i=s+1,\ldots, k$, $b_i\cdot b_j=1$ for $1\le i<j\le s$, and $b_i\cdot b_j=0$ for $i< j$, $j>s$. It turns out that
$$ G_1G^T_1=\left(\begin{array}{cc}
  J_s-I_s&O\\
  O&I_{k-s}\\
  \end{array}\right),$$
 where $G_1$ is the matrix whose rows are $b_1,\ldots,b_k$. The determinant of $G_1G_1^T$ is equal to $(-1)^{s-1}(s-1)$ and it is nonsingular, if $s$ is even. Let $s_i$ be the number of ones in the $i$-th column of the generator matrix $G$. If $s_i$ is even for some $i$, $1\le i\le n$ then the punctured code of $C$ on this coordinate is an LCD $[n-1,k,d$ or $d-1]$ code.

 If $s$ is odd, we consider the matrix $G_0$ with rows $b_2+b_1,\ldots$, $b_s+b_1$, $b_{s+1},\ldots,b_k$. This matrix generates the shortened code $C_0$. Now we have
 $$ G_0G^T_0=\left(\begin{array}{cc}
  J_{s-1}-I_{s-1}&O\\
  O&I_{k-s}\\
  \end{array}\right),$$
  and so det$G_0G_0^T=(-1)^{s-2}(s-2)$. Hence $G_0G_0^T$ is a nonsingular matrix and the shortened code $C_0$ is LCD. It turns out that  if $s_i$ is odd for some $i$, $1\le i\le n$ then the shortened code of $C$ on this coordinate is an LCD $[n-1,k-1,\ge d]$ code.

  If $s_i$ is even for all $i\in\{1,\ldots,k\}$ then $c_1+\cdots+c_k=0$ which is not the case, as these vectors are linearly independent. Hence there is at least one $i$ such that $s_i$ is odd and the corresponding shortened code is LCD. If $s_i$ is odd for all $i\in\{1,\ldots,k\}$ then $c_1+\cdots+c_k=(11\ldots 1)=\mathbf{1}$ which means that the shortened code of $C$ on the $i$-th coordinate is LCD for  all $i\in\{1,\ldots,k\}$, but none of the punctured codes is LCD.
\end{proof}

The above propositions give the following corollaries.

\begin{corollary}\label{cor:n-k-odd}
If $n-k$ is odd then $d_{LCD}(n,k)= d_{LCD}(n-1,k)$ or $d_{LCD}(n-1,k)+1$.
\end{corollary}

\begin{proof}
Since $n-k$ is odd then all LCD $[n,n-k]$ codes are odd-like. Hence no LCD $[n,k]$ code contains the all-ones vector. According to the above propositions, if $C$ is an optimal $[n,k,d]$ LCD code, then at least one of the punctured codes of $C$ on one coordinate position is an $[n-1,k,d$ or $d-1]$ LCD code. Hence $d_{LCD}(n-1,k)\ge d_{LCD}(n,k)-1$.
\end{proof}

\begin{corollary}\label{cor:odd1}
If $k$ is odd then $d_{LCD}(n,k)\le d_{LCD}(n-1,k-1)$.
\end{corollary}

\begin{proof}
Let $C$ be an LCD $[n,k,d]$ code. Since $k$ is odd then there is at least one coordinate position $i$ such that the shortened $[n-1,k-1,\ge d]$ code on the $i$-th coordinate is LCD. Hence $d_{LCD}(n,k)\le d_{LCD}(n-1,k-1)$.
\end{proof}

\begin{corollary}\label{cor:odd2}
If $k$ is even and $d_{LCD}(n,k)>d_{LCD}(n-1,k-1)$, then all LCD $[n,k,d_{LCD}(n,k)]$ codes are even.
\end{corollary}

\begin{proof}
Let $C$ be an $[n,k,d=d_{LCD}(n,k)]$ LCD code.
If $C$ is odd-like, then there is at least one coordinate position $i$ such that the shortened $[n-1,k-1,\ge d]$ code on the $i$-th coordinate is LCD. Hence $d_{LCD}(n-1,k-1)\ge d_{LCD}(n,k)$ which is not the case. Hence the code $C$ must be an even-like code.
\end{proof}

The following proposition shows the connection between LCD codes of length $n$ and those of length $n+2$ having a dual distance 2, with the same dimension $k$. It is a simple generalization of \cite[Proposition 2]{Massey}.

\begin{proposition}\label{prop:2}
Let $C$ be a binary $[n,k,d]$ code with a generator matrix $G$, $G'=(v^T \ v^T \ G)$ and $C'=\langle G'\rangle$, where $v\in\F_2^k$. The code $C$ is LCD if and only if $C'$ is an LCD code.
\end{proposition}

\begin{proof}
Let $x,y\in C$, $x'=(a,a,x)\in C'$, and $y'=(b,b,y)\in C'$, $a,b\in\F_2$. Then $x'\cdot y'=ab+ab+x\cdot y=x\cdot y$.
This turns out that $G'G'^T=GG^T$. It follows from Theorem \ref{thm:Massey} that if one of these codes is LCD then the other is also LCD.

It is worth mentioning that if the dual distance of $C$ is $d^\perp >1$ then the dual distance of $C'$ is equal to 2. So, if we take an LCD code with dual distance 2 and remove the corresponding equal coordinates from all the codewords, we will get an LCD code again.
\end{proof}

\begin{corollary}\label{cor:2}
If $d_{LCD}(n,k)$ is odd, then $d_{LCD}(n+2,k)\ge d_{LCD}(n,k)+1$.
\end{corollary}

\begin{proof}
Let $C$ be a binary LCD $[n,k,d=d_{LCD}(n,k)]$ code with a generator matrix $G$. If $v\in\F_2^k$ is the vector for which $(v G)$ generates an even $[n+1,k,d_1]$ code then $d_1=d+1$. Hence $G'=(v v G)$ generates an LCD $[n+2,k,\ge d+1]$ code and therefore $d_{LCD}(n+2,k)\ge d_{LCD}(n,k)+1$.
\end{proof}

We also present two corollaries that prove the existence of odd-like and even-like LCD codes of given length and dimension.

\begin{corollary}\label{cor:even}
Let $k$ be an even positive integer and $n>k$.
If $k\ge 4$ then there exist LCD $[n,k]$ even-like codes with $d^\perp\ge 2$. LCD $[n,2]$ even-like codes with dual distance $d^\perp\ge 2$ exist if and only if $n$ is odd, $n>1$.
\end{corollary}

\begin{proof}
Let $k=2m$. The parity check $[2m+1,2m,2]$ code is an LCD even-like code. Suppose that there exists an even-like LCD $[2m+2s+1,2m]$ code with dual distance at least 2. According to Proposition \ref{prop:2}, there is an even-like $[2m+2s+3,2m]$ LCD code with $d^\perp=2$. By mathematical induction an LCD even-like $[n,k]$ code for odd $n$ exists.

If $m\ge 2$, the $[2m+2,2,3]$ code $C=\langle (1110\ldots 0),(11\ldots 1)\rangle$ is an LCD code that contains the all-ones vector. Hence its dual code is an even-like LCD $[2m+2,2m]$ code with dual distance 3. Again, by mathematical induction and Proposition \ref{prop:2} we prove, that there are LCD even-like $[n,k]$ codes for all even $n>k$.

Let $k=2$ and $C=\langle \underbrace{1\ldots 1}_a \underbrace{1\ldots 1}_b \underbrace{0\ldots 0}_c,
\underbrace{1\ldots 1}_a \underbrace{0\ldots 0}_b \underbrace{1\ldots 1}_c\rangle$.  The code $C$ is even if and only if $a\equiv b\equiv c\pmod 2$. When $a$, $b$ and $c$ are even integers then the code is self-orthogonal, and when they are odd integers, $C$ is an LCD code. In the second case $n=a+b+c$ is odd.
\end{proof}

\begin{corollary}\label{cor:odd}
If $k\ge 2$ and $n\ge k+2$ are integers then there exist LCD $[n,k]$ odd-like codes with $d^\perp\ge 2$.
\end{corollary}

\begin{proof} The proof is similar to the proof of Corollary \ref{cor:even}. Let $k$ be an integer, $k\ge 2$. We first show that an LCD $[k+2,2,\ge 2]$ code $C$ with $\mathbf{1}\not\in C$  exists. Its dual code is an odd-like LCD $[k+2,k]$ code with dual distance $\ge 2$. Then by mathematical induction and Proposition \ref{prop:2} we prove, that there are LCD odd-like $[n,k]$ codes for all $n\ge k+2$.
\end{proof}

Some other constructions which allow modifying the parameters of codes and to obtain LCD codes from other codes, which can be LCD or not, are given in \cite{Carlet_Guilley}.

%%%%%%%%%%%%%%%%%%%   Section  %%%%%%%%%%%%%%%%%%%%
\section{Largest minimum weights for binary LCD codes}
\label{sect:dLCD}

The largest minimum weights $d_{LCD}(n, k)$ among all binary LCD $[n, k]$ codes were determined
in \cite{Kim-LCD}, \cite{Harada-LCD} and \cite{Araya-Harada} for $n \le 12$, $13\le n \le 16$, and $17\le n \le 24$, respectively. In this
section, we extend the results to lengths up to 40.

Obviously, the only linear $[n,n]$ code is $\F_2^n$ and it is an LCD code, so $d_{LCD}(n,n)=1$ for all $n\ge 1$.
According to \cite[Proposition 3.2]{Dougherty-Kim-Sole-LCD}, (i) if $n$ is odd, then $d_{LCD}(n, 1) = n$ and $d_{LCD}(n, n-1) = 2$, (ii) if $n$ is even, then $d_{LCD}(n, 1) = n-1$ and $d_{LCD}(n,n-1) = 1$.

It was shown in \cite{Kim-LCD} that
\[
d_{LCD}(n,2) = \left\{\begin{array}{ll}

\lfloor \frac{2n}{3}\rfloor& \mbox{if} \ n\equiv 1,2,3,4\pmod 6\\
\lfloor \frac{2n}{3}\rfloor-1& \mbox{otherwise},
\end{array}\right.
\]
for $n\ge 2$. The optimal LCD codes of dimension 2 have been classified in \cite{Harada-LCD}.

For $k=3$, Harada and Saito \cite{Harada-LCD} have proved that
\[
d_{LCD}(n,3) = \left\{\begin{array}{ll}

\lfloor \frac{4n}{7}\rfloor& \mbox{if} \ n\equiv 3,5\pmod 7\\
\lfloor \frac{4n}{7}\rfloor-1& \mbox{otherwise},
\end{array}\right.
\]
for $n\ge 3$.

In \cite{Araya-Harada} and \cite{Araya-Harada-Saito} the authors have established the following values of the function $d_{LCD}$:
\[
d_{LCD}(n,4) = \left\{\begin{array}{ll}

\lfloor \frac{8n}{15}\rfloor& \mbox{if} \ n\equiv 5,9,13\pmod{15}\\
\lfloor \frac{8n}{15}\rfloor-1& \mbox{if} \ n\equiv 1,2,3,4,6,7,8,10,11,12,14\pmod{15}\\
\lfloor \frac{8n}{15}\rfloor-2& \mbox{if} \ n\equiv 0\pmod{15},
\end{array}\right.
\]
for $n\ge 4$, and
\[
d_{LCD}(n,5) = \left\{\begin{array}{ll}

\lfloor \frac{16n}{31}\rfloor-1& \mbox{if} \ n\equiv 3,5,7,11,19,20,22,26\pmod{31}\\
\lfloor \frac{16n}{31}\rfloor-2& \mbox{if} \ n\equiv 4\pmod{31},
\end{array}\right.
\]
for $n\ge 5$.

Another general value of the considered function for codes with rate that tents to 1 is given in \cite[Proposition 3.2]{Kim-LCD}, namely
\[ d_{LCD}(n,n-i)=2 \ \ \mbox{for} \ \ i\ge 2, \ \ n\ge 2^i.\]

Instead of the given exact values of $d_{LCD}(n,k)$, we use some bounds. Obviously, if $C$ is an $[n,k,d]$ LCD code, the $[n+1,k,d]$ code $(0\vert C)$ is also LCD. Hence $d_{LCD}(n,k)\le d_{LCD}(n+1,k)$. Another important inequality is the following

\begin{theorem}{\rm \cite[Theorem 8]{CarletLCD2019}}\label{thm:k-1}
If $2\le k\le n$, then $d_{LCD}(n,k)\le d_{LCD}(n,k-1)$.
\end{theorem}

In addition to the inequalities, presented in Theorem \ref{thm:k-1} and the corollaries from the previous section, we use the bound $d_{LCD}(n,k)\le d(n,k)$, where $d(n,k)$ is the largest minimum weight $d$ for which a linear $[n,k,d]$ code exists. We look for the values of $d(n,k)$ in the Code Tables \cite{Grassl-table}.

Moreover, we have proved the nonexistence of LCD $[n,k,\ge d]$ codes for given integers $n$, $k$ and $d$. To do this, we use the program \textsc{Generation} of the software package \textsc{QExtNewEdition} \cite{Generation}. This program classifies the linear $[n,k,\ge d]$ codes with dual distance $\ge d^\perp$ for given integers $n$, $k$, $d$ and $d^\perp$. With \textsc{Generation} we have also classified or only constructed LCD codes with some parameters.

The program \textsc{Generation} has two main parts.
The first one implements a  construction method for  generator matrices.
This method is based on row by row backtracking  with $k \times k$ identity matrix as a  fixed part.
In the $m$-th step the considered matrices have the following form
$$G=(I_k \ A') = \left(\begin{array}{c|c|c}
I_{m}&O&A_m \\
\hline
O&I_{k-m}&X\\
\end{array}\right)$$
where the columns of the matrix $A_m$ are lexicographically ordered, and $X$ is the unknown part of $G$.
In that case any vector $v_m$ of length $n-k$ which fits for the $m$-th row of $A_m$ strictly depends on one of the  vectors put on the previous rows.
%Let $T_1$ be the set of all suitable rows for $A_1$. Take $v_1\in T_1$. The vector $v_1$ defines an ordered partition  $\Pi_{v_1}$ of the coordinate set $S=\{k+1,k+2, \ldots, n \}$.
%Let the matrix $A_{m-2}$ already be constructed. We define a set $T_{m-1}$ of all suitable vectors for the the last row in the next matrix $A_{m-1}$. Taking $v_{m-1}\in T_{m-1}$, we obtain the matrix $A_{m-1}$. The vector $v_{m-1}$ defines an ordered partition $\Pi_{v_{m-1}}$ of the set $S=\{k+1,k+2, \ldots, n \}$. The possibilities for the next $m$-th row correspond to the refinement partitions of $\Pi_{v_{m-1}}$ induced by the vectors in $T_{m-1}$.

The second part of the program is related to the identification of non-equivalent objects in the whole generation process. The general method which we apply is known as \textit{canonical augmentation} \cite{Kaski_Ostergard,McKay1998}. Description for this specific case is given in \cite{Generation}.
The basic idea is to accept only non-equivalent objects without an equivalence test (in some cases with a small number of tests) at every step of the generation process. Instead of an equivalence test, a canonical form of the objects and a canonical ordering of orbits are used.
So for every vector $v_m$ in the construction that fits as an $m$-th row (we call these vectors \textit{possible} solutions),  the algorithm decides acceptance (possible solution becomes \textit{real})  or rejection.
In this model, the different branches of the search tree are independent and therefore it is easy for parallel implementation.

The program \textsc{Generation} allows a lot of restrictions on the considered codes. To construct and classify LCD codes, we use the following proposition.

\begin{proposition}
Let $C$ be a binary linear $[n,k,d]$ code and $\dim (C\cap C^\perp)=s$. If $C_i$ is the shortened code of $C$ on $i$-th coordinate for some $i\in\{1,\ldots,n\}$ then $\dim (C_i\cap C_i^\perp)\le s+1$.
\end{proposition}

\begin{proof} Without loss of generality, we can consider the shortened code $C_1$ of $C$ on the first coordinate. Then the dual code of $C_1$ is the punctured code of $C^\perp$ on the same coordinate. If all codewords have 0 as a first coordinate then the shortened and the punctured codes on this coordinate coincide, otherwise $C=(0\vert C_1)\cup (1\vert x+C_1)$ for a codeword $(1,x)\in C$. Let $\mathcal{H}=C\cap C^\perp$ and  $\mathcal{H}_1=C_1\cap C_1^\perp$. There are two possibilities for $\mathcal{H}$, namely $\mathcal{H}=(0\vert \mathcal{H}')$ or $\mathcal{H}=(0\vert \mathcal{H}')\cup (1\vert v+\mathcal{H}')$. In both cases $\mathcal{H}'\subseteq \mathcal{H}_0$. If $\mathcal{H}'= \mathcal{H}_1$ then $\dim \mathcal{H}_1=\dim \mathcal{H}$ or $\dim \mathcal{H}-1$.

Let now $\mathcal{H}'\not\equiv \mathcal{H}_1$.
 Take $y_1,y_2\in \mathcal{H}_1\setminus \mathcal{H}'$. It follows that $(0,y_i)\in C$ and $(1,y_i)\in C^\perp$ for $i=1,2$. Hence $(0,y_1+y_2)\in \mathcal{H}$ and so $y_1+y_2\in \mathcal{H}'$. It turns out that $\mathcal{H}_1=\mathcal{H}'\cup (y_1+\mathcal{H}')$ and $\dim \mathcal{H}_1=\dim \mathcal{H}'+1$. Since $\dim \mathcal{H}'=\dim \mathcal{H}$ or $\dim \mathcal{H}-1$, we have $\dim \mathcal{H}_1=\dim \mathcal{H}$ or $\dim \mathcal{H}+1$.
\end{proof}

During the generation process, in the $m$-th step, the program \textsc{Generation} constructs all linear $[n-k+m+1,m+1]$ codes that have the already constructed $[n-k+m,m]$ codes as shortened codes. To obtain LCD codes, we are looking for the dimension of the hulls of the codes in any step. For example, if we would like to obtain all $[9,4,4]$ LCD codes, we construct during the generation process all $[8,3,4]$ codes whose hulls have dimensions $\le 1$, $[7,2,4]$ codes with $\dim \mathcal{H}\le 2$ and $[6,1,\ge 4]$ codes with $\dim \mathcal{H}\le 3$.

%\begin{proposition}\textrm{[\cite{Dougherty-Kim-Sole-LCD},Lemma 3.1]}
% For $n$ and $k$ integers greater than $0$, $d_{LCD}(n + 1, k) \ge d_{LCD}(n, k)$.
%\end{proposition}

To obtain the exact values of $d_{LCD}(n,k)$, we use the following inequalities:
\begin{equation}\label{eq1}
d_{LCD}(n-1,k)\le d_{LCD}(n,k)\le d(n,k).
\end{equation}
\begin{equation}\label{eq2}
\mbox{If} \ k \ \mbox{is odd then} \ d_{LCD}(n,k)\le d_{LCD}(n-1,k-1).
\end{equation}
\begin{equation}\label{eq3}
\mbox{If} \ d_{LCD}(n-2,k) \ \mbox{is odd then} \ d_{LCD}(n-2,k)+1\le d_{LCD}(n,k)\le d(n,k).
\end{equation}
If $k$ is even and $d_{LCD}(n-1,k)$ is odd then
\begin{equation}\label{eq4}
d_{LCD}(n-1,k)+1\le d_{LCD}(n,k)\le d(n,k).
\end{equation}

In the cases when these inequalities are not enough, we use the program \textsc{Generation}.

\subsection{Length 25}

As the values of $d_{LCD}(25,k)$ are known for $k\le 4$ and $k\ge 21$, we consider the dimensions $k\in\{5,\ldots,20\}$.

According to (\ref{eq1}), $d_{LCD}(25,9)=d_{LCD}(25,10)=8$, $d_{LCD}(25,13)=6$, and $d_{LCD}(25,16)=d_{LCD}(25,17)=d_{LCD}(25,18)=4$. As $d_{LCD}(23,7)=9$ and $d(25,7)=10$ then $d_{LCD}(25,7)=10$. According to (\ref{eq4}) $d_{LCD}(25,14)=6$, since $d_{LCD}(24,14)=5$ and $d(25,14)=6$.

Using the program \textsc{Generation}, we proved that no LCD $[25,5,12]$ code exists but there are 122 inequivalent LCD $[25,5,11]$ codes with dual distance $d^\perp\ge 2$. This result gives us that $d_{LCD}(25,5)=11$. Two of the constructed codes have dual distance 3 which means that $d_{LCD}(25,20)=d(25,20)=3$. Using the same program we prove that any LCD $[25,20,3]$ code is equivalent to one of these two codes. During the generation process in the case $n=25$, $k=5$, $d=12$, we obtain the following information:
\begin{itemize}
\item there are 24 $[22,2,\ge 12]$ codes $C_2$ such that $\dim(C_2\cap C_2^\perp)\le 3$;
\item there are 13 $[23,3,\ge 12]$ codes $C_3$ such that $\dim(C_3\cap C_3^\perp)\le 2$;
\item there are two $[24,4,\ge 12]$ codes $C_4$ such that $\dim(C_4\cap C_4^\perp)\le 1$ - one of them is the only LCD $[24,4,12]$ code and the other one is the only $[23,4,12]$ code extended with a zero coordinate (it is not LCD).
\end{itemize}

The program \textsc{Generation} shows that no LCD $[25,6,11]$, $[25,11,8]$ and $[25,19,4]$ codes exist. Combining this result with (\ref{eq1}) we obtain $d_{LCD}(25,6)=10$, $d_{LCD}(25,11)=7$ and $d_{LCD}(25,19)=3$. Applying Theorem \ref{thm:k-1} we have $d_{LCD}(25,12)\le 7$. There are four LCD $[25,12,7]$ codes. All these codes have dual distance 6.

In the cases $k=8$ and $k=15$ we have constructed more than 100000 LCD $[25,8,9]$ and more than 165000 LCD $[25,15,5]$ codes.

\subsection{Length 26}

The values of $d_{LCD}(26,k)$ are known for $k\le 5$ and $k\ge 22$. Using the inequalities (\ref{eq1})-(\ref{eq4}), we have $d_{LCD}(26,7)=d_{LCD}(26,8)=10$, $d_{LCD}(26,10)=d_{LCD}(26,11)=d_{LCD}(26,12)=8$, $d_{LCD}(26,14)=6$, and $d_{LCD}(26,k)=4$ for $k=17,18,19,20$.

Using the program \textsc{Generation}, we proved that no LCD $[26,6,12]$ exists but there are 221 inequivalent LCD $[26,6,11]$ codes. There are more than 4000 LCD $[26,9,9]$ and more than 2000 LCD $[26,16,5]$ codes. There is only one LCD $[26,21,3]$ code and its dual distance is 11.

Consider the case $k=13$ in more detail. Gulliver and \"Osterg\aa rd have determined that there are three optimal inequivalent [26,13,7] codes, $C_{26,1}$, $C_{26,2}$, and $C_{26,3}$ (see \cite{GO} for their generator matrices and weight enumerators). Exactly two of them, namely $C_{26,2}$ and $C_{26,3}$, are LCD codes. Their dual distances are equal to $7$. We list here only their weight enumerators.
%\[G_1=\left(\begin{array}{c}
%00011111111111000000000000\\
%11100000011110100000000000\\
%00100001101110010000000000\\
%00100110010110001000000000\\
%01111000001010000100000000\\
%10011001100010000010000000\\
%10001010001110000001000000\\
%01101101100010000000100000\\
%01010110100010000000010000\\
%10010111000100000000001000\\
%11101100110100000000000100\\
%10100011011000000000000010\\
%11010010111000000000000001\\
%\end{array}\right), \ \ \
%%AUT: 78 dd -7
%G_2=\left(\begin{array}{c}
%01111111111111000000000000\\
%10000001111110100000000000\\
%00000110011110010000000000\\
%00001011100110001000000000\\
%10110010011010000100000000\\
%00110110101000000010000000\\
%01110001001010000001000000\\
%11011101110010000000100000\\
%01011000111000000000010000\\
%11001010001110000000001000\\
%11100110100100000000000100\\
%00011101010100000000000010\\
%11100101011100000000000001\\
%\end{array}\right).\]
\begin{align*}
  W_{26,1}= & ~1+117z^{7}+273z^{8}+338z^{9}+598z^{10}+923z^{11}+1105z^{12}+1340z^{13} \\
   & +1300z^{14}+923z^{15}+598z^{16}+338z^{17}+182z^{18}+117z^{19}+39z^{20},\\
  W_{26,3}=&~1+115z^{7}+275z^{8}+350z^{9}+586z^{10}+893z^{11}+1135z^{12}+1380z^{13}\\
  &+1260z^{14}+893z^{15}+628z^{16}+350z^{17}+170z^{18}+115z^{19}+41z^{20}.
\end{align*}

The two codes are not orthogonal to each other, therefore both are isodual. This shows that there are linear codes which are both LCD and isodual.

Only the case $k=15$ remains unsolved. We failed to construct or prove the nonexistence of LCD $[26,15,6]$ codes and so $d_{LCD}(26,15)=5$ or 6.

\subsection{Length 27}

The values of $d_{LCD}(27,k)$ are known for $k\le 4$ and $k\ge 23$. From the inequalities (\ref{eq1})-(\ref{eq4}) we obtain $d_{LCD}(27,5)=d_{LCD}(27,6)=12$, $d_{LCD}(27,8)=10$, $d_{LCD}(27,11)=d_{LCD}(27,12)=8$, $d_{LCD}(27,15)=d_{LCD}(27,16)=6$, and $d_{LCD}(27,18)=d_{LCD}(27,19)=d_{LCD}(27,20)=4$.

Using \textsc{Generation}, we proved the nonexistence of $[27,9,10]$, $[27,13,8]$, $[27,21,4]$ and $[27,22,3]$ LCD codes. This gives us that  $d_{LCD}(27,9)=9$, $d_{LCD}(27,13)=7$, $d_{LCD}(27,21)=3$ and $d_{LCD}(27,22)=2$.

We obtained exactly 33 $[27,7,11]$ LCD codes. Moreover, we constructed four $[27,10,9]$ LCD codes, and two of them have dual distance 5 which proves that $d_{LCD}(27,17)=5$ as $d(27,17)=5$ \cite{Grassl-table}.  According to (\ref{eq4}), if $d_{LCD}(27,14)=7$ then $d_{LCD}(28,14)=8$ but the program \textsc{Generation} completed the case for the LCD $[28,14,8]$ codes with the result that such LCD codes do not exist. Therefore $d_{LCD}(27,14)=6$.

\subsection{Length 28}

The values of $d_{LCD}(26,k)$ are known for $k\le 4$ and $k\ge 24$. Using the inequalities (\ref{eq1})-(\ref{eq4}), we have $d_{LCD}(28,6)=12$, $d_{LCD}(28,9)=d_{LCD}(28,10)=10$, $d_{LCD}(28,11)=d_{LCD}(28,12)=d_{LCD}(28,13)=8$,
$d_{LCD}(28,15)=d_{LCD}(28$, $16)=6$, and $d_{LCD}(28,k)=4$ for $k=19,20,21$.

Using the program \textsc{Generation}, we proved that no $[28,7,12]$, $[28,8,11]$, $[28,14,8]$ and $[28,22,4]$ LCD codes exist, and $d_{LCD}(28,7)=11$, $d_{LCD}(28,8)=10$, $d_{LCD}(28,14)=7$, and $d_{LCD}(28,22)=3$. We obtained that there are two LCD $[28,5,13]$ and five $[28,10,10]$ codes. Four of the codes of dimension 10 have dual distance 5, hence $d_{LCD}(28,18)=5$.

Instead of using the program \textsc{Generation}, we prove a proposition which is important for the larger lengths, too.

\begin{proposition}\label{prop:dd2}
If $n\ge 4$, $k\le n-2$, and $d_{LCD}(n,k)=2$, then $d_{LCD}(n+i,k+i)=2$ for all integers $i\ge 1$.
\end{proposition}

\begin{proof}
Suppose that $d=d_{LCD}(n+1,k+1)>2$ and $C$ is an LCD $[n+1,k+1,d]$ code. According to Corollaries \ref{cor:odd1} and \ref{cor:odd2}, $k+1$ must be even and $C$ must be an even-like code, so $d\ge 4$. According to Proposition \ref{prop:punctured} the punctured code of $C$ on any coordinate is an LCD $[n,k+1,d']$ with $d'\ge d-1\ge 3$. Since $d_{LCD}(n,k+1)\le d_{LCD}(n,k)=2$ (see Theorem \ref{thm:k-1}), we have a contradiction. Hence $d\le 2$.

According to Corollaries \ref{cor:even} and \ref{cor:odd} there is an LCD $[n+1,n-k]$ code with  $d^\perp\ge 2$. This means that an LCD $[n+1,k+1,\ge 2]$ code exists. Hence $d=d_{LCD}(n+1,k+1)=2$.
\end{proof}

Only the case $k=17$ remains unsolved. We failed to construct or prove the nonexistence of LCD $[28,17,6]$ codes and so $d_{LCD}(28,17)=5$ or 6.

\subsection{Length 29}

The values of $d_{LCD}(27,k)$ are known for $k\le 4$ and $k\ge 25$. For the other values of $k$, we calculated $d_{LCD}(29,k)$ using the inequalities and the program \textsc{Generation}. Only the exact value of $d_{LCD}(29,11)$ remains unknown.

\subsection{Length 30}

In this case $d_{LCD}(30,1)=29$, $d_{LCD}(30,2)=19$, $d_{LCD}(30,3)=16$, $d_{LCD}(30,4)$ $=14$,\\  $d_{LCD}(30,26)=d_{LCD}(30,27)=d_{LCD}(30,28)=2$, and $d_{LCD}(30,29)=d_{LCD}(30,30)=1$. From the inequalities given above, we have $d_{LCD}(30,7)=d_{LCD}(30,8)=12$, $d_{LCD}(30,13)=d_{LCD}(30,14)=8$, $d_{LCD}(30,17)=d_{LCD}(30,18)=6$, $d_{LCD}(30,21)=d_{LCD}(30,22)=4$.

According to Corollary \ref{cor:odd1}, $d_{LCD}(30,5)\le 14$ and $d_{LCD}(30,9)\le 11$. Using the program \textsc{Generation}, we obtained exactly two inequivalent $[30,5,14]$ LCD codes and a few $[30,9,11]$ LCD codes. Moreover, we proved the nonexistence of LCD $[30,6,14]$, $[30,10,11]$ and $[30,24,4]$ codes, and construct LCD $[30,6,13]$, $[30,10,10]$, $[30,20,5]$, $[30,23,4]$ and $[30,24,3]$ codes.

For some of the remaining cases for $k$, namely $k\in\{ 11,12, 15,16,19\}$ we have:
\begin{itemize}
\item[$k=15$)] Suppose that $C$ is an LCD $[30,15,8]$ code. According to Theorem \ref{thm:even}, $C$ and $C^\perp$ are odd-like codes. Hence the all-ones vector is not a codeword of $C$ and according to Proposition \ref{prop:odd-like}, there is a coordinate $i$ such that the punctured code of $C$ on this coordinate is a $[29,15,7$ or $8]$ LCD code. As $d_{LCD}(29,15)=6$, this is not possible. It follows that $d_{LCD}(30,15)\le 7$.
\item[$k=16$)] If $C$ is an LCD $[30,16,7]$ code then it is odd-like and according to Proposition \ref{prop:odd-like} at least one of the shortened codes of $C$ of length $n-1$ is LCD. Hence there is an LCD $[29,15,\ge 7]$ code which contradicts to $d_{LCD}(29,15)=6$. Hence $d_{LCD}(30,16)=6$.
\end{itemize}

\subsection {Optimal LCD Codes of Lengths Greater Than 30}

A summary of the results given in this and the previous subsections
is presented in Table \ref{table1} and Table \ref{table2}. For lengths $n \leq 40$ and
dimensions $5 \leq k \leq 32$, the maximal minimum distance of corresponding codes is shown.

We constructed even-like $[33,10,12]$, $[33,12,10]$, $[37,12,12]$, $[35,14,10]$  and $[37,16,10]$ LCD codes. These constructions prove that $d_{LCD}(33,12)\ge 10$, $d_{LCD}(37,12)\ge 12$, $d_{LCD}(35,14)$ $=10$ and $d_{LCD}(37,16)=10$. According to Proposition \ref{prop:punctured}, there are $[32,12,9]$, $[36,12,11]$, $[34,14,9]$  and $[36,16,9]$ LCD codes and so $d_{LCD}(32,12)\ge 9$, $d_{LCD}(36,12)\ge 11$, $d_{LCD}(34,14)\ge 9$, and $d_{LCD}(36,16)\ge 9$. Moreover, $d_{LCD}(33,13)\ge 9$, $d_{LCD}(34,13)\ge 9$, $d_{LCD}(35,15)$ $\ge 9$, $d_{LCD}(36,15)\ge 9$. Puncturing the constructed 1304 even-like $[37,16,10]$ LCD codes on the last coordinate, we obtain the same number of odd-like $[36,16,9]$ LCD codes. Neither of them contains the all-ones vector and therefore some of the punctured $[35,16,8$ or $9]$ and some of the shortened $[35,15,\ge 9]$ codes are LCD. Hence $d_{LCD}(35,16)\ge 8$ and $d_{LCD}(35,15)\ge 9$. Similarly we construct odd-like $[34,15,8]$ and $[34,16,7]$ LCD codes that are not self-complementary.

From the constructed two odd-like LCD $[35,11,12]$ codes we obtain that $d_{LCD}(35,11)=12$, $d_{LCD}(34,11)\ge 11$, $d_{LCD}(34,12)\ge 11$, $d_{LCD}(35,12)\ge 11$, and $d_{LCD}(36,11)\ge 12$. If $d_{LCD}(34,12)=12$ then according to Corollary \ref{cor:odd2}, all LCD $[34,12,12]$ codes are even and so their $[33,12,\ge 11]$ punctured codes are LCD, which contradicts $d_{LCD}(33,12)=10$. Hence $d_{LCD}(34,12)=11$, $d_{LCD}(35,12)=12$, and $d_{LCD}(36,12)=12$.

Moreover, we proved that odd-like $[31,9,12]$ and even-like $[32,10,12]$ LCD codes do not exist. This result gives us that $d_{LCD}(31,9)=d_{LCD}(32,10)=11$ (according to the presented above bounds). Using Corollaries \ref{cor:k-even-d-odd} and \ref{cor:odd1}, we have $d_{LCD}(31,10)=d_{LCD}(32,11)=d_{LCD}(33,12)=10$.
\begin{table}
\caption{Bounds on the minimum distance of binary LCD codes\label{table1}}
{\small
\begin{tabular}{r|c|c|c|c|c|c|c|c|c|c|c|c|c }
\hline\noalign{\smallskip}
$n\backslash k$ & 5 &6 &7 & 8 &9 &10 &11 &12&13 &14 &15 &16  &17 \\
\noalign{\smallskip}\hline\noalign{\smallskip}
16 &6   &6*  &5   &5*    & 4* & 4*   & 3    & 2*   & 2*& 2*  & 1   & 1    &           \\
17 &7   &6   &6*  &6*    & 5* & 4*   & 3   & 3*   & 2*& 2*  & 2*  & 2*   & 1         \\
18 &7   &7   &6   &6*    & 5  & 4*   & 4*   & 4*   &3* & 2*  & 2*  & 2*   & 1         \\
19 &8*  &8*  &7   &6     & 6* & 5*   & 4*   & 4*   & 3 & 3*  & 2*  &2*    & 2*        \\
20 &9*  &8*  &7   &6     & 6  & 6*   & 5*   & 4*   & 4*& 4*  & 3*  &2*    & 2*       \\ \hline
21 &9   &8*  &8*  &7     & 6  & 6    & 5    & 5*   & 4*& 4*  & 3   &3*    & 2*     \\
22 &10* &9*  &8*  &8*    & 7  & 6    & 6    & 6*   & 5*& 4*  & 4*  &4*    & 3*       \\
23 &10  &10* &9*  &8*    & 7  & 7    & 6    & 6    & 5 & 4   & 4*  &4*    & 3      \\
24 &11  &10* &9   &8*    & 8* & 8*   & 7    & 6    & 6*& 5   & 4*  &4*    & 4*      \\
25 &11  &10  &10* &9*    & 8* & 8*   & 7    & 7    & 6*& 6*  & 5*  &4*    & 4*     \\
26 &12* &11  &10  &10*   & 9* & 8*   & 8*   & 8*   & 7*& 6*  & 5--6& 5*   & 4*    \\
27 &12  &12* &11  &10*   & 9  & 9*   & 8*   & 8*   & 7 & 6   & 6*  & 6*   & 5*      \\
28 &13  &12* &11  &10    & 10*& 10*  & 8*   & 8*   &8* & 7   & 6*  & 6*   & 5--6        \\
29 &13  &12  &12* &11    & 10 & 10*  & 8--9 & 8*   & 8*& 8*  & 6   & 6*   & 6*         \\
30 &14  &13  &12* &12*   & 11 & 10   & 8--10& 8--9 & 8*& 8*  & 6--7& 6    & 6*        \\
\hline
%                                10      11      12     13     14      15     16   17
31 &14  &14  &13* &12* &11      &10     &9--10  &8--10 &8--9  & 8*   &6--8  &6--7 &6          \\
32 &15  &14  &13  &12  &12*     &11     &10     &9--10 &8--10 &8--9  &6--8  &6--8 &6--7         \\
33 &15  &14  &14* &13  &12*     &12*    &10--11 &10    &9--10 &8--10 &7--9  &6--8 &6--8         \\
34 &16* &15  &14  &14* &12--13  &12*    &11--12 &11    &9--10 &9--10 &8--10 &7--9&6--8         \\
35 &16* &16* &15  &14  &12--14  &12--13 &12*    &12*   &10--11 &10*   &9--10 &8--10&6--8     \\
36 &17* &16* &15  &14  &12--14  &12--14 &12--13 &12*   &10--12 &10--11 &9--10 &9--10&6--9     \\
37 &17  &16  &16* &15  &12--14  &12--14 &12--14 &12--13 &10--12 &10--12 &10--11 &10*&6--10        \\
38 &18* &17  &16* &16* &12--15  &12--14 &12--14 &12--14 &10--12 &10--12 &10--12 &10--11&6--10        \\
39 &18  &18* &17* &16* &12--16  &12--15 &12--14 &12--14 &10--13 &10--12 &10--12 &10--12&6--11   \\
40 &19  &18* &17  &16* &12--16  &12--16 &12--15 &12--14 &10--14 &10--13 &10--12 &10--12&6--12    \\
\noalign{\smallskip}\hline
\end{tabular}}
\end{table}

\begin{table}
\caption{Bounds on the minimum distance of binary LCD codes\label{table2}}
{\small
\begin{tabular}{r|c|c|c|c|c|c|c|c|c|c|c|c|c|c|c}
\hline\noalign{\smallskip}
$n\backslash k$
   & 18 &19   &20  & 21    &22  &23    &24    &25  &26 &27 &28 &29  &30 &31& 32  \\
   \noalign{\smallskip}\hline\noalign{\smallskip}
%  &$d$ & \# &$d$ & \#      & $d$& \#     & $d$& \#   &$d$& \#     &$d$& \# &$d$& \#     &$d$& \#      \\
23 &3*  &2*   &2*  &2*    & 2* & 1    &     &     &  &    &   &  &     & &    \\
24 &4*  &3*   &2*  &2*    & 2* & 1    & 1   &     &  &    &   &  &     & &     \\
25 &4*  &3    &3*  &2*    & 2* & 2*   & 2*  & 1   &  &    &   &  &     & &    \\
26 &4*  &4*   &4*  &3*    & 2* & 2*   & 2*  & 1   & 1&    &   &  &     & &     \\
27 &4*  &4*   &4*  &3     & 2  & 2*   & 2*  & 2*  &2*& 1  &   &  &     & &     \\
28 &5*  &4*   &4*  &4*    & 3  & 2    & 2*  & 2*  &2*& 1  & 1 &  &     & &      \\
29 &6*  &5*   &4*  &4*    & 4* & 3    & 2   & 2*  &2*& 2* & 2*& 1&     & &      \\
30 &6*  &5--6 &5*  &4*    & 4* & 4*   & 3   & 2   &2*& 2* & 2*& 1& 1   & &     \\
\hline
%   18    19    20    21     22     23   24    25    26   27   28   29   30  31    32
31 &6*   &6*   &6*   &4--5 &4*   &4*   &4*   &3    &2    & 2* &2*  &2*  &2* & 1  &       \\
32 &6    &6*   &6*   &4--6 &4--5 &4*   &4*   &3--4 &3    &2*  &2*  &2*  &2* & 1  &1     \\
33 &6--7 &6    &6*   &4--6 &4--6 &4--5 &4*   &4*   &4*   &3*  &2*  &2*  &2* & 2* &2*     \\
34 &6--8 &6--7 &6    &4--6 &4--6 &4--6 &4*   &4*   &4*   &3--4&3*  &2*  &2* & 2* & 2*      \\
35 &6--8 &6--8 &6--7 &4--6 &4--6 &4--6 &4--5 &4*   &4*   &4*  &4*  &3*  &2* & 2* & 2*    \\
36 &6--8 &6--8 &6--8 &4--7 &4--6 &4--6 &4--6 &4--5 &4*   &4*  &4*  &3--4&3* & 2* & 2*    \\
37 &6--9 &6--8 &6--8 &4--8 &4--7 &4--6 &4--6 &4--6 &4--5 &4*  &4*  &4*  &4* & 3*   & 2*     \\
38 &6--10&6--9 &6--8 &4--8 &4--8 &4--7 &4--6 &4--6 &4--6 &4--5&4*  &4*  &4* & 3--4   & 3*      \\
39 &6--10&6--10&6--9 &4--8 &4--8 &4--8 &4--7 &4--6 &4--6 &4--6&4--5&4*  &4* & 4*   & 4*  \\
40 &6--11&6--10&6--10&4--9 &4--8 &4--8 &4--8 &4--7 &4--6 &4--6&4--6&4--5&4* & 4*   & 4*   \\
\noalign{\smallskip}\hline
\end{tabular}}
\end{table}

%%%%%%%%%%%%%%%%%%%   Section  %%%%%%%%%%%%%%%%%%%%
\section{Some classification results of optimal LCD codes}
\label{sect:classification}

Using the program \textsc{Generation}, we have classified optimal binary LCD codes for some values of $n$ and $k$. The results are presented in Table \ref{tableso}. Any column consists of two subcolumns. We present the value of $d_{LCD}(n,k)$ for the corresponding length $n$ and dimension $k$ in the first subcolumn, and the number of inequivalent LCD $[n,k,d_{LCD}(n,k)]$ with $d^\perp\ge 2$ in the second one.

We have classified all optimal LCD $[n,k,d]$ codes of dimensions 4 and 5 for $17\le n\le 30$. For dimension 6, only the case $n=29$ is not completed, since the $[29,6,12]$ LCD codes are more than a million. There are more gaps in the table for codes with larger dimensions since the number of inequivalent LCD codes with corresponding parameters increases dramatically.

Looking at the results in Table \ref{tableso}, we made the observation that if $d_{LCD}(n,k)$ is even and $d_{LCD}(n-1,k)=d_{LCD}(n,k)-1$ for an even positive integer $k$, then all LCD $[n,k,d_{LCD}(n,k)]$ are even-like codes. We are not sure if this is true in all cases but we present it as a conjecture.

\begin{conjecture} Let $k$ be an even positive integer and $n>k$ be another integer. If $d_{LCD}(n,k)$ is even and $d_{LCD}(n-1,k)=d_{LCD}(n,k)-1$, then all LCD $[n,k,d_{LCD}(n,k)]$ are even-like codes.
\end{conjecture}

In addition, we have classified optimal LCD codes of dimension $k\ge 11$ and/or $n\ge 30$ for a few values of $n$ and $k$. The results are listed in Table \ref{tableso2}.

Files with generator matrices and some additional information about the constructed codes, the classification and non-existence results, are presented in \cite{Stefka-Zenodo}.

\begin{table}
\caption{Classification of binary LCD codes $(d^\perp\ge 2)$\label{tableso}}
{\small
\begin{tabular}{r|cl|cl|cl|cl|cl|cl|cl}
\hline\noalign{\smallskip}
$n\backslash k$
   & 4 &    & 5 &     & 6  &      & 7 &      & 8 &      & 9 &      & 10&  \\
   \noalign{\smallskip}\hline\noalign{\smallskip}
17 &8  &2   &7  &10   & 6  & 3187 & 6 & 7    & 6 & 1    & 5 & 1    & 4 &4550     \\
18 &8  &20  &7  &495  & 7  & 11   & 6 &18604 & 6 &337   & 5 &11857 & 4 &         \\
19 &9  & 2  &8  &20   & 8  & 2    & 7 & 3    & 6 &299426& 6 & 3    & 5 &11554     \\
20 &10 &1   &9  &1    & 8  & 392  & 7 &73511 & 6 &      & 6 &      & 6 & 601       \\
\hline
21 &10  & 10 &9  & 72  & 8  & 48671  & 8  & 19   & 7 & 312215 & 6 &   & 6 &     \\
22 &10  & 76&10 & 1   & 9  &  105   & 8  &626306& 8 & 14330  & 7 &231430   & 6 &       \\
23 &11  & 2 &10 &104  & 10 & 14     & 9  & 11 & 8 &     & 7 &     & 7 &        \\
24 &12  &1  &11 & 1   & 10 & 4434   & 9 &1952980     & 8 &    & 8 &    & 8 &  358    \\
\hline
25 &12  &11 &11 & 122 & 10 &735670 & 10 & 49& 9 &   & 8 &   & 8 &       \\
26 &12  &106&12 & 1   & 11 & 221   & 10 &   &10 &216675   & 9 &   & 8 &     \\
27 &13  &9  &12 & 173 & 12 & 24    & 11 & 33&10 &   & 9 &   & 9 &    \\
28 &14  &2  &13 &  2  & 12 &21136  & 11 &   &10 &   & 10&   & 10& 5          \\
29 &14  &33 &13 & 477 & 12 &       & 12 &   & 11&   & 10&   & 10&     \\
30 &14  &310&14 &  2  & 13 &1024   & 12 &   & 12&   & 11&   &10 &         \\
\noalign{\smallskip}\hline
\end{tabular}}
\end{table}

\begin{table}
\caption{Classification of binary LCD codes $(d^\perp\ge 2)$\label{tableso2}}
{\small
\begin{tabular}{r|c|c|c|c|c|c|c|c}
\hline\noalign{\smallskip}
$[n,k,d]$
    &[17,11,3] &[18,11,4]  & [17,12,3]    &[18,12,4]  &[19,12,4]    &[20,14,4]   &[22,16,4]  &[24,18,4]    \\\hline
%  &$d$ & \# &$d$ & \#      & $d$& \#     & $d$& \#   &$d$& \#     &$d$& \# &$d$& \#     &$d$& \#      \\
  &26020   &6003 &31    & 6 & 13139   & 4    &  3   & 1      \\
   \noalign{\smallskip}\hline\noalign{\smallskip}
$[n,k,d]$
    &[31,5,14] &[31,6,14]  & [31,7,13]    &[32,5,15]  &[35,6,16]    &[40,5,19]   &[30,23,4]  &    \\\hline
%  &$d$ & \# &$d$ & \#      & $d$& \#     & $d$& \#   &$d$& \#     &$d$& \# &$d$& \#     &$d$& \#      \\
  &608   &59 &73    & 2 & 102   & 17    &  15   &       \\
\noalign{\smallskip}\hline
\end{tabular}}
\end{table}

\section{Conclusion}
\label{sec_conclusion}

In this paper, we extend the study of binary LCD codes and provide tables for the largest possible minimum weight $d_{LCD}(n,k)$ for the LCD $[n,k]$ codes. Looking at these tables, we made some observations that we would like to point out.

\begin{itemize}
\item There are LCD $[2k,k]$ codes that are also isodual (and formally self-dual). For example, the only $[8,4,3]$ LCD code, the only $[16,8,5]$ LCD code, and both $[26,13,7]$ LCD codes. It is easy to see, that any LCD formally self-dual code must be odd-like (otherwise the hull will contain the all-ones vector). A more systematic study of these codes would be interesting.
\item In many cases, the optimal LCD $[n,k]$ codes have large dual distance, moreover, their dual codes are also optimal LCD codes. For example, four of the optimal $[28,10,10]$ codes have dual distance 5 (and therefore their dual codes are optimal $[28,18,5]$ codes), and the last one has dual distance 4. In the same time, if $d_{LCD}(n+1,k)=d_{LCD}(n,k)$, the optimal $[n,k]$ codes with added zero coordinate to all codewords, produce optimal $[n+1,k]$ codes with dual distance 1. The problem is: Under what conditions
are an LCD code and its dual optimal at the same time.
%\item  We have already mentioned that according our classification results, if $d_{LCD}(n,k)$ is even and $d_{LCD}(n-1,k)=d_{LCD}(n,k)-1$ for an even positive integer $k$, then all LCD $[n,k,d_{LCD}(n,k)]$ are even-like codes. This is another open problem for optimal LCD codes.
\item The computation results show that $$d_{LCD}(n+1,k)=d_{LCD}(n,k) \ \mbox{or} \ d_{LCD}(n,k)+1.$$ We conjecture that this is true for all lengths $n$ and dimensions $k$.
\end{itemize}

\section*{Acknowledgements}
I would like to thank Petar Boyvalenkov for focusing my attention to LCD codes and introducing me to the main literature related to this type of codes. I am also grateful to Iliya Bouyukliev for the included restriction about LCD codes in his program \textsc{Generation}.




\end{document}